\documentclass[12pt,draftcls,onecolumn]{IEEEtran}
\IEEEoverridecommandlockouts
\usepackage{cite}
\usepackage{graphicx}      

\usepackage[utf8]{inputenc}
\usepackage[T1]{fontenc}
\usepackage{amsthm}

\newcommand{\mS}{\mathcal{S}}
\newcommand{\mN}{\mathcal{N}}
\newcommand{\mP}{\mathcal{P}}

\newcommand{\mE}{\mathcal{E}}
\newcommand{\mF}{\mathcal{F}}

\newcommand{\mD}{\mathcal{D}}
\newcommand{\mG}{\mathcal{G}}

\newcommand{\mL}{\mathcal{L}}

\newcommand{\mV}{\mathcal{V}}
\newcommand{\mU}{\mathcal{U}}

\newtheorem{theorem}{\bf Theorem}

\newtheorem{proposition}{\bf Proposition}
\newtheorem{lemma}{\bf Lemma}

\newtheorem{assumption}{\bf Assumption}

\newtheorem{definition}{\bf Definition}
\newtheorem{remark}{\bf Remark}

\newcommand{\vnorm}[1]{\| #1 \|}

\usepackage{color}
\definecolor{Red}{rgb}{1,0,0}
\definecolor{Blue}{rgb}{0,0,1}
\definecolor{Green}{rgb}{0,1,0}
\definecolor{magenta}{rgb}{1,0,.6}
\definecolor{lightblue}{rgb}{0,.5,1}
\definecolor{lightpurple}{rgb}{.6,.4,1}
\definecolor{gold}{rgb}{.6,.5,0}
\definecolor{orange}{rgb}{1,0.4,0}
\definecolor{hotpink}{rgb}{1,0,0.5}
\definecolor{newcolor2}{rgb}{.5,.3,.5}
\definecolor{newcolor}{rgb}{0,.3,1}
\definecolor{newcolor3}{rgb}{1,0,.35}
\definecolor{darkgreen1}{rgb}{0, .35, 0}
\definecolor{darkgreen}{rgb}{0, .6, 0}
\definecolor{darkred}{rgb}{.75,0,0}

\newcommand{\bigO}[1]{\mathcal{O} \left(#1 \right)}
\newcommand{\ts}{t^\star}

\newcommand{\real}{{\mathbb{R}}}

\newcommand{\integers}{\mathbb{Z}}

\newcommand{\realnonnegative}{{\mathbb{R}}_{\ge 0}}

\usepackage{amsmath}

\usepackage{amssymb}
\usepackage{algorithm}
\usepackage{algorithmic}
\usepackage{array}
\usepackage{comment}
\usepackage{url}
\usepackage{amsmath}

\DeclareMathOperator*{\argmax}{arg\,max}
\newcommand\subscr[2]{#1_{\textup{#2}}}

\begin{document}

\title{Robust Distributed Averaging: When are Potential-Theoretic Strategies Optimal?}
\author{Ali Khanafer and Tamer Ba\c{s}ar
\thanks{A preliminary version of this work will was presented at the IEEE Conference on Decision and Control (CDC), Florence, Italy 2013 \cite{KhanaferTouriBasarCDC13}. This work was supported in part by an AFOSR MURI Grant FA9550-10-1-0573.}
\thanks{Ali Khanafer and Tamer Ba\c{s}ar are with the Coordinated Science Laboratory, ECE Department, University of Illinois at Urbana-Champaign, USA {\tt\small {khanafe2,basar1}@illinois.edu}}
}

\maketitle

\begin{abstract}
We study the interaction between a network designer and an adversary over a dynamical network. The network consists of nodes performing continuous-time distributed averaging. The adversary strategically disconnects a set of links to prevent the nodes from reaching consensus. Meanwhile, the network designer assists the nodes in reaching consensus by changing the weights of a limited number of links in the network. We formulate two Stackelberg games to describe this competition where the order in which the players act is reversed in the two problems. Although the canonical equations provided by the Pontryagin's maximum principle seem to be intractable, we provide an alternative characterization for the optimal strategies that makes connection to potential theory. Finally, we provide a sufficient condition for the existence of a saddle-point equilibrium for the underlying zero-sum game.
\end{abstract}


\section{Introduction}
Various physical and biological phenomena where global patterns of behavior stem from local interactions have been modeled using linear distributed averaging dynamics. In such dynamics an agent updates its value as a linear combination of the values of its neighbors. Averaging dynamics is the basic building block in many multi-agent systems, and it is widely used whenever an application requires multiple agents, who are graphically constrained, to synchronize their measurements. Examples include formation control, coverage, distributed estimation and optimization, and flocking \cite{SaberMurray,SaberFlocking,NedicOzdaglarParrilo}. Besides engineering, linear distributed averaging finds applications in other fields as well. For instance, social scientists use averaging to describe the evolution of opinions in networks \cite{JacksonGolub}.

In practice, communication among agents is prone to different types of non-idealities which can affect the convergence properties of the associated distributed algorithms. Transmission delays \cite{NedicOzdaglar}, noisy links \cite{XiaoBoyd,TouriNedic}, and quantization \cite{KashyapBasarSrikant} are some examples of non-idealities that are due to the physical nature of the application. In addition to physical restrictions, researchers have also studied averaging dynamics in the presence of malicious nodes in the network \cite{Sundaram,SandbergJohansson}. 

In \cite{KhanaferTouriBasarNecSys12}, we explored the effect of an external adversary who attempts to prevent the nodes from reaching consensus by launching network-wide attacks. When the adversary is capable of disconnecting certain links in the network, we derived the optimal strategy of the adversary and demonstrated that it admits a potential-theoretic analogy. In this paper, we also introduce a network designer that attempts to counter the effect of the adversary and help the nodes reach consensus. The designer is capable of changing the weights of certain links. Both the adversary and the designer are constrained by their physical capabilities, e.g., battery life and communication range. To capture such constraints, we allow the adversary and the designer to affect only a fixed number of links at any point in time. The conflicting objectives of the designer and the adversary calls for a game-theoretic formulation to study their interaction.

Such a competition between a network designer and an adversary can occur in various practical applications. For example, in a wireless network, the link weights in such a network represent the capacities of the corresponding links. The designer can modify the capacity of a certain link using various communication techniques such as introducing parallel channels between two nodes as in orthogonal frequency division multiple access (OFDMA) networks \cite{tse2005fundamentals}. {In OFDMA networks, the number of parallel links between two nodes is usually limited \cite{remy2014lte}. To capture this limitation, we limit the amount by which the designer can increase the capacity of a given link. The adversary can be a jammer who is capable of breaking links by injecting high noise signals that disrupt the communication among nodes. The adversary is assumed to have sufficient transmit power to disrupt the communication over any link, no matter what the number of parallel channels is.} 

Our model in this paper is different from the models in the current literature in two ways: (i) the adversary and the designer compete over a dynamical network. This is different from the problems studied in the computer science and economics communities where the network is usually static \cite{Goyal2010Robust}; (ii) the players in our model are constrained and do not have an infinite budget. This enables us to model practical scenarios more closely rather than allowing the malicious behavior to be unrestricted as in \cite{LeBlanc,BicchiBullo,Sundaram}, {where it is assumed that the network contains nodes that are misbehaving. In addition, those papers focus on finding necessary and sufficient conditions for the network to reach consensus in the presence of malicious nodes, and observability theory is the main tool used to study such problems. Here, we assume that all the nodes are normal, and we focus on identifying the links that are of importance to the adversary and the designer who have conflicting objectives. This requires us to borrow tools from differential games and optimal control theory.} 

The main goal of this work is to derive optimal strategies for the designer and the adversary. {By modeling the adversary as a strategic player and deriving optimal defense strategies, we guarantee robustness against worst-case attacks, unlike existing approaches in which attacks on links were modeled as random failures \cite{kar2010distributed}}. Because the order in which the players act affects the resulting utilities, we formulate two Stackelberg games based on the order of play, allowing a different player to have the \emph{first-move-advantage}. When the adversary is allowed to play first, he is capable of restricting the available actions of the designer since some links will disappear from the network. Hence, if we were to cast the problem as a zero-sum game between the players, we should not expect the existence of a saddle-point equilibrium (SPE) in pure strategies. The question we would like to answer is then: \emph{are there scenarios where the order of play does not affect the eventually attained utilities of the players, which leads to the existence of an SPE?}

Accordingly, the contributions of this paper are as follows. We capture the interaction between the designer and the adversary by formulating two separate problems. In the min--max problem, the designer declares a strategy first to which the adversary reacts by its optimal response. The second problem is a max--min one where the order of play is reversed. Assuming the controllers do not switch infinitely many times over a finite interval among the available actions, we derive the optimal strategies for both problems in terms of potential-theoretic quantities by working directly with the utility functional. Furthermore, we demonstrate that the derived strategies satisfy the necessary conditions provided by the maximum principle (MP). Finally, we derive a sufficient condition guaranteeing the existence of an SPE.

The rest of the paper is organized as follows. In Section \ref{ProbDesc} we describe the min--max and max--min problems. In Section \ref{Problems}, we derive the Stackelberg strategies and show that they satisfy the MP. We provide a sufficient condition for the existence of an SPE in Section \ref{suffCond}. We end the paper with concluding remarks and delineation of future research directions of Section \ref{Conclusion}. An Appendix includes a proof of one of the theorems and two technical results.

\subsection*{Notation and Terminology} 
We will use $\sum_{j > i} (.)$ to mean $\sum_{j=2}^{n}\sum_{i=1}^{j-1}(.) $, $[.]^T$ to denote the transpose of a vector or a matrix $[.]$, and $\mathbf{1}$ to denote the $n$-dimensional column vector of $1$'s. The Euclidean norm of a vector is denoted by $\|.\|_2$ and the $\ell_1$-norm of a vector is denoted by $\|.\|_1$. {The absolute value of a scalar variable is denoted by $|.|$, which we also use to denote the cardinality of a set---the intended use of this operator will be clear from the context.} The $(i,j)$-th element of a matrix $X$ is denoted by $X_{ij}$. We will often use $x$ to refer to a function or its value at a given time instant; the context should make the distinction clear. We will use the words ``strategy" and ``action" interchangeably; since we are seeking optimal open-loop strategies in this paper, both terms are equivalent.
{
A graph is a pair $\mG = (\mN,\mE)$, where $\mN$ is the set of nodes, and $\mE \subseteq \mN \times \mN$ is the set of edges. An edge from node $i\in \mN$ to node $j\in \mN$ is denoted by $e_{ij}$, i.e., $e_{ij}:=(i,j)$. A graph is called undirected if $e_{ij} \in \mE$ if and only if $e_{ji} \in \mE$. A path is a collection of nodes $\{i_1,\hdots, i_l \} \subseteq \mN$, $l \in  \integers_{> 1} $, such that $e_{i_k i_{k+1}} \in \mE$ for all $k \in \{1,\hdots,l-1 \}$. We call an undirected graph \emph{connected} if it contains a path between any two nodes in $\mN$.
}
{
Given an undirected graph $\mG = (\mN,\mE)$, we define the projection operator $\Phi: \mE \times \real \to \mE$ such that $\Phi((e,r)) = e$, for some $(e,r) \in \mE \times \real$. When applied to a set $S \subset \mE \times \real$, the mapping $\Phi$ is defined as follows:
\begin{eqnarray*}
\Phi(S) & = & \left\{
  \begin{array}{l l}
    \bigcup\limits_{(e,r)\in S} \Phi((e,r)), &  \text{$S \neq \emptyset$}  \\
    0, &  \text{$S = \emptyset$}
  \end{array} \right.
\end{eqnarray*}
Given $S \subset \mE \times \real$, with $|S| = k$, let $\pi(S)=\{(e_1,r_1),\hdots,(e_k,r_k) \}$, where $r_i\in \real$ and $e_i\in \mE$ for all $i\in \{1,\hdots,k\}$, be an ordering of the elements of $S$ such that $r_1 \leq \hdots \leq r_k$. Then, given $\ell \in \integers_{\geq 0}$, we define the set operator $\Phi_\ell: \mE\times \real \to \mE$ as:
\begin{eqnarray*}
\Phi_\ell(S) & = & \left\{
  \begin{array}{l l}
    \Phi(S) , &  \text{$\ell > k$}  \\
    \{e_1,\hdots,e_\ell \}, &  \text{$0<\ell \leq k$}\\
    0, & \text{$\ell = 0$ or $k=0$}
  \end{array} \right.
\end{eqnarray*}
Throughout the paper, we will be dealing with undirected graphs. Although both $e_{ij},e_{ji}$ belong to the set of edges $\mathcal{E}$ in such graphs, we do not distinguish between the two edges, and we treat them as a single edge. As a result, in any set defined over $\mE\times \real$, we include a \emph{single} tuple $(e_{ij},r_{ij})$, $r_{ij} \in \real$, to represent both edges. 
}

\section{Problem Formulation} \label{ProbDesc}
Consider a connected network of $n$ nodes and $m$ links described by a weighted undirected graph $\mathcal{G} = (\mN, \mE)$. The value, or state, of the nodes at time instant $t\in \realnonnegative$ is given by $x(t) = [x_1(t),...,x_n(t)]^T$. The nodes start with an initial value $x(0)=x_0$, and they are interested in computing the average of their initial values, $\subscr{x}{avg} = \frac{1}{n}\sum_{i=1}^n x_i(0)$, via local averaging. We consider the continuous-time averaging dynamics given by
\begin{equation} \label{systemEqn}
\dot{x}(t) = Ax(t), \quad x(0) = x_0,
\end{equation}
where the matrix $A$, $A_{ij} = a_{ij} \in \real$, has the following properties:
\begin{eqnarray*} 
A = A^T & , &   \quad A\mathbf{1} = 0, \\
A_{ij} \geq 0 & , & \quad A_{ij} = 0 \iff e_{ij} \notin \mE, \quad  i \neq j.
\end{eqnarray*}
Define $\bar{x} = \mathbf{1}\subscr{x}{avg} \in \mathbb{R}^n$ and let $M = \frac{1}{n}\mathbf{1}\mathbf{1}^T$. A well-known result states that, under the above assumptions, the nodes will reach consensus as $t\to \infty$, i.e.,  $\lim_{t\to \infty} x(t) =\bar{x}$ \cite{SaberMurray}. To achieve their respective objectives, the designer and the adversary control the elements of $A$ as we describe next. This will render the matrix $A$ to be time-varying.

The adversary attempts to slow down convergence by breaking at most $\ell \leq m$ links at each time $t$. Let $u_{ij}(t) \in \{0,1\}$ be the weight the adversary assigns to link $e_{ij} \in \mE$ at time $t \realnonnegative$. He breaks link $e_{ij}$ when $u_{ij}(t) = 1$. Define $r := {n\choose2}$. The action set of the adversary is then
\begin{eqnarray*}
U & = & \{w \in \mathbb{R}^r: w =[w_{12},...,w_{1n},w_{23},...,w_{(n-1)n}]^T, w_{ij} \in \{0,1\}, \\
&& w_{ij}=0  \text{ if } e_{ij}\notin \mE, \|w\|_1 \leq \ell \}. 
\end{eqnarray*}
The set of admissible controls, $\mU$, consists of all functions that are piecewise continuous in time and whose range is $U$. Given a time interval $[0,T]$, we can formally write
\begin{equation*}
\mU = \left\{u:[0,T] \to U \hspace{1mm} | \text{ $u$ is a piecewise continuous function of $t$} \right\}. \\
\end{equation*}

We introduce a network designer who attempts to accelerate convergence by controlling the weights of the edges. The designer can change the weight of a given link by adding $v_{ij}(t)$ to its weight $a_{ij}$. We assume that $v_{ij}(t) \in \{0,b\}$ and that the number of links the designer modifies is at most $\ell \leq m$. Given the above definitions, we can write down the $(i,j)$-th element, $i\neq j$, of the matrix $A(u(t),v(t))$ as
\begin{equation} \label{eqn::Laplace}
A_{ij}(u(t),v(t)) = (a_{ij} + v_{ij}(t)) (1-u_{ij}(t)), \quad \text{for all } e_{ij}\in \mE
\end{equation}
We require that the resulting matrix is a negative Laplacian of the graph; hence, we must have $A_{ii}(u(t),v(t)) = -\sum_{j \neq i} A_{ij}(u(t),v(t))$, for all $i \in \mV$. 

Given a time interval $[0,T]$, define the following functional:
\begin{equation*}
J(u,v) = \frac{1}{2}\int_0^T k(t)\vnorm{ x(t)-\bar{x}}_2^2 dt,
\end{equation*}
where the weighting factor $k(t)$ is positive and integrable over $[0,T]$, {which can, for example, be viewed as a discounting factor, such as $k(t) = e^{-\alpha t}$ for some $\alpha > 0$.} This constitutes the utility function of the adversary, and that of the designer is $-J(u,v)$. We will study two problems. In the first one, the adversary acts first by selecting the links he is interested in breaking. Then, the network designer optimizes his choices over the resulting graph, which we denote by $\mG(u(t)) = (\mN, \mE(u(t)))$, where $\mE(u(t)) = \mE \setminus \{e_{ij} \in \mE : u_{ij}(t)=1\}$. In this case, the action set of the designer can be written as
\begin{eqnarray*}
V(u(t)) & = & \left\{w\in \mathbb{R}^r: w =[w_{12},...,w_{1n},w_{23},...,w_{(n-1)n}]^T, w_{ij} \in \{0,b \},\right. \\
&& \left. w_{ij}=0  \text{ if } e_{ij}\notin \mE(u(t)), \|w\|_1 \leq b \ell  \right\}.
\end{eqnarray*}
The set of admissible controls for the designer, $\mV(u)$, consists of all piecewise continuous functions whose range is $V(u)$. Formally, we define
\begin{equation*}
\mV(u) = \left\{v:[0,T] \to V(u(t)) \hspace{1mm} | \text{ $v$ is a piecewise continuous function of $t$} \right\}. \\
\end{equation*}

The max--min problem can now be formally written as\footnote{Even though existence of a maximum and a minimum has not yet been shown at this stage, we will still call this the ``max--min" problem in anticipation of such an existence result later in the paper. The formal definition below is still in terms of sup and inf. The same argument applies to the min--max problem to be introduced shortly.}
\begin{eqnarray*}
&\sup\limits_{u \in \mathcal{U}} \inf\limits_{v \in \mathcal{V}(u)} & J(u,v) \\
& \text{subject to} & \dot{x}(t)=A(u(t),v(t))x(t), \quad x(0) = x_0.
\end{eqnarray*}

In the second problem, the order is reversed. Because the designer acts first in this problem, he can optimize over the entire graph $\mG$. Thus, the action set of the designer in this problem is $V := V(0)$ and the set of its admissible controls is $\mV := \mV(0)$; the sets of actions and admissible controls of the adversary remain the same. We can then write
\begin{eqnarray*}
&\inf\limits_{v \in \mathcal{V}} \sup\limits_{u \in \mathcal{U}}  & J(u,v) \\
& \text{subject to} & \dot{x}(t)=A(u(t),v(t))x(t), \quad x(0) = x_0.
\end{eqnarray*}

In a computer network, the max--min problem allows the network designer (who is the maximizer here) to architect networks that are robust against strategic virus diffusion. The min--max problem finds applications in army combat situations where the designer (the minimizer) attempts to counter the attacks of the enemy intending to disrupt the network communication. 

Given the nature of the players' possible modifications of the network, as described by \eqref{eqn::Laplace}, we can view the actions of the players as switches among the possible Laplacian matrices resulting from modifying the links. Moreover, the capability of the designer and the adversary to change the system matrix renders it as a ``switched" one. The optimal controllers for such systems can exhibit Zeno effect, i.e., they may switch infinitely many times over a finite interval. In order to explicitly eliminate the possibility of infinite switching, we make the following assumption in the remainder of this paper.
\begin{assumption} \label{assume::costSwitch}
Let $0\leq r_1< \hdots < r_{K_u}$ be the switching times of some $u \in \mU$ and $0 \leq s_1< \hdots < s_{K_v}$ be those of some $v \in \mV$. We assume that $K_u,K_v \in \integers_{\geq0}$ are finite, and that there exists a globally minimum dwell time $\tau> 0$ such that
\begin{equation}  \label{eqn::dwellTime}
\tau \leq \min \left\{r_{i+1}-r_i, s_{i+1}-s_i, |r_i-s_j|: 1\leq i \leq K_u, 1\leq j < K_v\right\},
\end{equation}
over which the system matrix $A(u,v)$ is time-invariant.
\end{assumption}
{Note that this assumption is well motivated for practical reasons. Consider,
for example, a communication network where an adversary is a jammer injecting an interfering signal at some links. If the adversary chooses to change the set of links it is jamming, there must be some delay for the adversary to change its configuration.} Now, we make the following assumption for both problems:
\begin{assumption} \label{assume::commonInfo}
The initial matrix $A(0,0)$, the time interval $[0,T]$, the values $\ell$ and $b$, and the initial state $x_0$ are common information to both players.
\end{assumption}

We recall the definition of an SPE.
\begin{definition}[Saddle-Point Equilibrium (SPE) \hspace{-0.1mm}\cite{BasarOlsder}] \label{def::existSPE}
The pair $(u^\star,v^\star)$ constitutes an SPE if it satisfies the following pair of inequalities
\begin{equation}
J(u,v^\star) \leq J(u^\star,v^\star) \leq J(u^\star,v), \label{eqn::SPE}
\end{equation}
for $u\in \mU,$ $v \in \mV$.
\end{definition}
The following remarks are now in order.
\begin{remark}
(Non-Rectangular Strategy Sets and Existence of SPE) When the strategy sets are rectangular, i.e., the strategy of one player does not restrict the strategy space of the other, the following relationship holds:
\begin{equation} \label{eqn::upperlowerV}
\underline{V} = \sup\limits_{u \in \mathcal{U}} \inf\limits_{v \in \mathcal{V}} J(u,v) \leq \inf\limits_{v \in \mathcal{V}}  \sup\limits_{u \in \mathcal{U}} J(u,v) = \overline{V},
\end{equation}
where $\underline{V}, \overline{V}$ are called, respectively, the lower and upper values of the game.
When the strategy sets are non-rectangular, however, the order in (\ref{eqn::upperlowerV}) may not hold. Moreover, one should not expect the pair of inequalities (\ref{eqn::SPE}) to hold, and hence an SPE may not exist. In the max--min problem in this paper, the strategy sets of the players are non-rectangular as the adversary's action, removing links from $\mG$, could restrict the actions available to the designer. 

\end{remark}
\begin{remark} \label{rem::complex}
(Problem Complexity) Let us consider the problem of the adversary for a given strategy of the designer. Assume that the adversary can act at $K_u \in \integers_{\geq 0}$ given time instances over the interval $[0,T]$. Then, for $\ell \leq m$, assuming that $\|u(t)\|_1 =\ell$ for all $t \in \realnonnegative$, the total number of links that need to be tested in a brute-force approach is
\begin{equation}\label{bruteForce}
{m\choose{\ell}}^{K_u} \geq \left(\frac{m}{\ell}\right)^{\ell K_u}.
\end{equation}
Clearly, the brute-force approach leads to an exponential number of computations as a function of $K_u$. The same argument applies to the problem faced by the network designer.
\end{remark}
\section{Optimal Strategies} \label{Problems}
We will now present the solutions to the two problems introduced above. In \cite{KhanaferTouriBasarCDC13}, we have shown that the canonical equations provided by the maximum principle (MP) are intractable due to the interdependence between the state, costate, and the optimal controls; therefore, it may not be possible to obtain the optimal strategies in closed form using the MP. Here, we take an alternative route to arrive at the optimal strategies of the players by working directly with the objective functional. In what follows, we will often drop the time index and other arguments for notational simplicity. We will be using the term ``connected component" to refer to a set of connected nodes which have the same values. The following quantities will be central to the derivation of the optimal strategies:
\begin{equation} \label{eqn::potTheory}
 \nu_{ij} := -(x_i-x_j)^2, \quad w_{ij} := (a_{ij}+v_{ij})\nu_{ij}.
\end{equation}

\subsection{The Min--Max Problem}
The following theorem presents the optimal strategy of the adversary in the min--max problem. Define the set 
\begin{equation} \label{eqn::setL}
\mL_\ell(v) = \Phi_\ell\left(\{(e_{ij},(a_{ij}+v_{ij})\nu_{ij}): e_{ij}\in \mE \}\right) \subseteq \mE.
\end{equation} 
 
\begin{theorem} \label{thm::implement}
Under Assumptions \ref{assume::costSwitch} and \ref{assume::commonInfo}, and for a fixed strategy $v$ of the designer, the optimal strategy of the adversary in the min--max problem is
\begin{eqnarray*}
u^\star_{ij}(v) & = & \left\{
  \begin{array}{l l}
    1, &  \text{$e_{ij} \in \mL_{\ell}(v)$}  \\
    0, &  \text{$e_{ij} \notin  \mL_{\ell}(v)$} 
  \end{array} \right.
\end{eqnarray*}
If the adversary has an optimal strategy of breaking fewer than $\ell$ links, then either $\mathcal{G}$ has a cut of size less than $\ell$ or the nodes have reached consensus by time $t$. In either of these cases, breaking $\ell$ links is also optimal.
\end{theorem}

\begin{proof}
For a fixed strategy of the designer $v \in \mV$, we will show that it is optimal for the maximizer to rank the links based on their $w_{ij}$ values, where $w_{ij}$ was defined in (\ref{eqn::potTheory}). Under Assumption \ref{assume::costSwitch}, the function $x$ becomes piecewise continuous. Hence, the function $w_{ij}$, for all $e_{ij} \in \mE$, is also piecewise continuous and its value cannot change abruptly over a finite interval. As a result, we can regard the system as a time-invariant one over a small interval $[t_0,t_0+\delta] \subset [0,T]$, where $0 < \delta \leq \tau$, and $\tau$ was defined in (\ref{eqn::dwellTime}). The proof consists of two steps.
\begin{enumerate}
\item Showing that, over a small interval $[t_0,t_0+\delta]$, it is optimal for the adversary to switch from a strategy $u \in \mU$ to another strategy $u^\star\in \mU$, where $u^\star$ entails breaking the $\ell$ links with the lowest $w_{ij}$ values.
\item Showing that allowing $u^\star$ to mimic $u$ for the remaining time of the problem preserves the gain obtained over $[t_0,t_0+\delta]$.
\end{enumerate}

Over a small interval, $u$ and $u^\star$ induce certain system matrices. Let the system matrix corresponding to $u$ over $[t_0,t_0+\delta]$ be $A(u,v)=A$, and let $\|u\|_1 < \ell$ over this interval. Because the control strategies of both players are time-invariant over this interval, we have
\begin{equation} \label{eqn::stateFixed}
x(t) = e^{A(t-t_0)}x(t_0), \quad t \in [t_0,t_0+\delta].
\end{equation}
Let $P(t):=e^{At}$. Due to the structure of $A$, $P(t)$ is a doubly stochastic matrix for $t \geq 0$ \cite[p. 63]{Norris}. Note that we can write $x(t_0) = \tilde{P}x_0$, where $\tilde{P}$ is some doubly stochastic matrix. Indeed, assume that either or both controls had switched once at some time $\tilde{t}_0 \in [0,t_0)$, and that the system matrix over $[0,\tilde{t}_0)$ was $\tilde{A}_1$, and the system matrix corresponding to $[\tilde{t}_0,t_0)$ was $\tilde{A}_2$. Then $x(t_0) = e^{\tilde{A}_2(t_0-\tilde{t_0})}e^{\tilde{A}_1\tilde{t_0}}x_0$. Because both $e^{\tilde{A}_1t}$, $e^{\tilde{A}_2t}$ are doubly stochastic matrices, their product is also doubly stochastic. We can readily generalize this result to any number of switches in the interval $[0,t_0)$. With this observation, we can write
\[
x(t)-\bar{x} =P(t-t_0)\tilde{P}x_0-Mx_0 = (P(t-t_0) - M)x(t_0),
\]
where the last equality follows from the fact that
\begin{equation} \label{prop::M}
\tilde{P}M = M\tilde{P} = M,\text{ $\tilde{P}$ is doubly stochastic}.
\end{equation}
We want to show that switching from strategy $u$ to strategy $u^\star$ at some time $t^\star \in [t_0,t_0+\delta]$, can improve the utility of the adversary. To this end, we assume that the matrix induced by $u^\star$ over $[t_0,t^\star)$ is $A$, while the system matrix corresponding to $u^\star$ over $[t^\star,t_0+\delta]$ is $B$. Define the doubly stochastic matrix $Q(t) := e^{Bt}$, $t\geq 0$. Over $[t^\star,t_0+\delta]$, the strategies $u$ and $u^\star$ are identical except at link $e_{ij} \in \mE$, where $u_{ij} = 0$ and $u_{ij}^\star = 1$, i.e., $\|u\|_1 < \|u^\star\|_1 $ over this sub-interval. It follows that:
\begin{equation}
A_{ij} > B_{ij}=0, \quad A_{kl} = B_{kl} \quad \forall e_{kl} \neq e_{ij}.  \label{eqn::AvsB}
\end{equation}
Formally, we want to prove the following inequality:
\begin{eqnarray*}
\int_{t_0}^{t_0+ \delta} &k(t)&\vnorm{(P(t-t_0)-M)x(t_0)}_2^2dt  \\
< \int_{t_0}^{t^\star} &k(t)&\vnorm{(P(t-t_0)-M)x(t_0)}_2^2dt \nonumber \\
& + & \int_{t^\star}^{t_0+\delta} k(t)\vnorm{(Q(t-t^\star)-M)P(t^\star - t_0)x(t_0)}_2^2dt, \nonumber
\end{eqnarray*}
or equivalently
\begin{eqnarray}
\int_{t^\star}^{t_0+\delta} k(t)\cdot \left[\vnorm{(Q(t-t^\star)-M)P(t^\star-t_0)x(t_0)}_2^2\right. \nonumber\\
- \left. \vnorm{(P(t-t_0)-M)x(t_0)}_2^2\right] dt > 0. \label{eqn::intermediate}
\end{eqnarray}
Using (\ref{prop::M}) and the semi-group property, (\ref{eqn::intermediate}) simplifies to
\begin{equation}
\int_{t^\star}^{t_0+\delta} k(t)\cdot x(t_0)^T\Lambda(t,t^\star)x(t_0) dt > 0, \label{ineq::toProve}
\end{equation}
where $\Lambda(t,t^\star)= P(t^\star-t_0)Q(2(t-t^\star))P(t^\star-t_0) - P(2(t-t_0))$.
A sufficient condition for (\ref{ineq::toProve}) to hold is
\begin{equation} \label{eqn::suffToProve}
h(t,x(t_0)) = x(t_0)^T\Lambda(t,t^\star)x(t_0)>0, \text{ for } t> t^\star.
\end{equation}
As $\delta \downarrow 0$, we can write $P(t) = I + tA+\bigO{\delta^2}$, where $\bigO{\delta^2}/\delta \leq L$ for sufficiently small $\delta$ and some finite constant $L$. We therefore have
\begin{eqnarray}\label{eqn::matApprox}
&& \Lambda(t,t^*) = \left(I+(t^\star-t_0) A+\bigO{\delta^2}\right)\left(I+2(t-t^\star)B \right. \nonumber \\
&&\left.+\bigO{\delta^2}\right) \left(I+(t^\star-t_0) A+\bigO{\delta^2}\right) - \left(I+2(t-t_0)A \right. \nonumber \\
&&\left. +\bigO{\delta^2}\right) = 2(t-t^\star)B + 2(t^\star-t_0)A - 2(t-t_0)A  \nonumber \\
&&+ \bigO{\delta^2}  = 2(t-t^\star)(B-A) + \bigO{\delta^2}. 
\end{eqnarray}
For sufficiently small $\delta$, the first term dominates the second term. Recall that the quadratic form of a Laplacian matrix $L$ exhibits the following form: $x^TLx = \sum_{l=1}^n \sum_{k=1}^{l-1} L_{kl}(x_l-x_k)^2$, for any $x \in \mathbb{R}^n$. Note that $B-A$ is in fact a negative Laplacian. Using (\ref{eqn::AvsB}), we can then write
\begin{eqnarray}
h(t,x(t_0)) & = & 2(t-t^\star)  \sum_{r>s}(A_{sr} - B_{sr})\left(x_r(t_0)-x_s(t_0)\right)^2  + \bigO{\delta^2}\nonumber \\
& = & 2(t-t^\star)A_{ij}\left(x_j(t_0)-x_i(t_0) \right)^2 + \bigO{\delta^2}. \label{eqn::condInit}
\end{eqnarray}
For small enough $\delta$, the higher order terms are dominated by the first term. Hence, if there is a link $e_{ij}$ such that $x_i(t_0) \neq x_j(t_0)$, there exists $t^\star$ such that $h(t,x(t_0))>0$ for $t \in \left(t^\star,t_0+\delta\right]$. Since $t_0$ was arbitrary, we conclude that the optimal strategy must satisfy $\|u^\star(t)\|_1 = \ell$ for all $t$, given that each of the $\ell$ links connects two nodes having different values. 

If no link such that $x_i(t_0) \neq x_j(t_0)$ exists at a given time $t_0$, the adversary does not need to break additional links, although breaking more links does not affect optimality because $h(t,x(t_0))=0$ in such a case. There are two cases under which the adversary cannot find a link to make $h(t,x(t_0))>0$: (i) The graph at time $t_0$ is one connected component. In this case, the nodes have already reached consensus and $\|u^\star(t)\|_1< \ell$. This is a \emph{losing strategy} for the adversary as he has failed in preventing nodes from reaching agreement; (ii) The graph at time $t_0$ has multiple connected components, and the number of links connecting the components is less than $\ell$. The adversary here possesses a \emph{winning strategy} with $\|u^\star(t)\|_1 < \ell$, as he can disconnect $\mathcal{G}$ into multiple components and prevent consensus. 

Next, we need to show that the adversary will modify the $\ell$ links with the lowest $w_{ij}$ values. Let us again restrict our attention to the interval $[t_0,t_0+\delta]$ where the adversary applies strategy $u$. Assume (to the contrary) that the links the adversary breaks over this interval are not the ones with the lowest $w_{ij}$ values. In particular, assume that the adversary chooses to break link $e_{kl}$, while there is a link $e_{ij}$ such that $w_{ij} < w_{kl}$. Assume that the adversary switches at time $t^\star \in [t_0,t_0+\delta]$ to strategy $u^\star$ by \emph{breaking} link $e_{ij}$ and \emph{unbreaking} link $e_{kl}$. Then, (\ref{eqn::condInit}) becomes
\[
h(t,x(t_0))=2(t-t^*)\left(w_{kl}(t_0)-w_{ij}(t_0)  \right) + \bigO{\delta^2}.
\]
Hence, by following the same arguments as above, we can conclude that breaking $e_{kl}$ is not optimal. 

The second step of the proof is to show that switching to strategy $u^\star$ guarantees an improved utility for the adversary \emph{regardless of how the original trajectory corresponding to $u$ changes beyond time $t_0+\delta$}. To this end, we will assume that from time $t_0+\delta$ onward, strategy $u^\star$ will \emph{mimic} strategy $u$. Assume that strategy $u$ switches from matrix $A$ to matrix $C$ over the interval $[t_0+\delta,t_0+2\delta]$, and define $R(t) := e^{Ct}$. Hence, strategy $u^\star$ will also switch from the system matrix $B$ to matrix $C$. However, the trajectories corresponding to $u$ and $u^\star$ will have different initial conditions at time $t_0+\delta$, due to the switch that strategy $u^\star$ made at time $t^\star$. Fig. \ref{fig::proofSketch} illustrates this idea. Recall that according to $A$, we have $\|u\|_1 < \ell$ and $u_{ij} = 0$. Here, the system matrix $B$ can differ from the matrix $A$ in two ways: either (i) $B$ dictates breaking one additional link compared to $A$, or (ii) $B$ dictates breaking link $e_{ij}$ and unbreaking link $e_{kl}$ where $w_{ij} < w_{kl}$.  Consider Case (i) first and let us study the behavior of the system over the interval $[t_0+\delta, t_0 + 2\delta]$ where we can assume that the system is time-invariant.
\begin{figure}[!t]
\centering
\includegraphics[width=8cm]{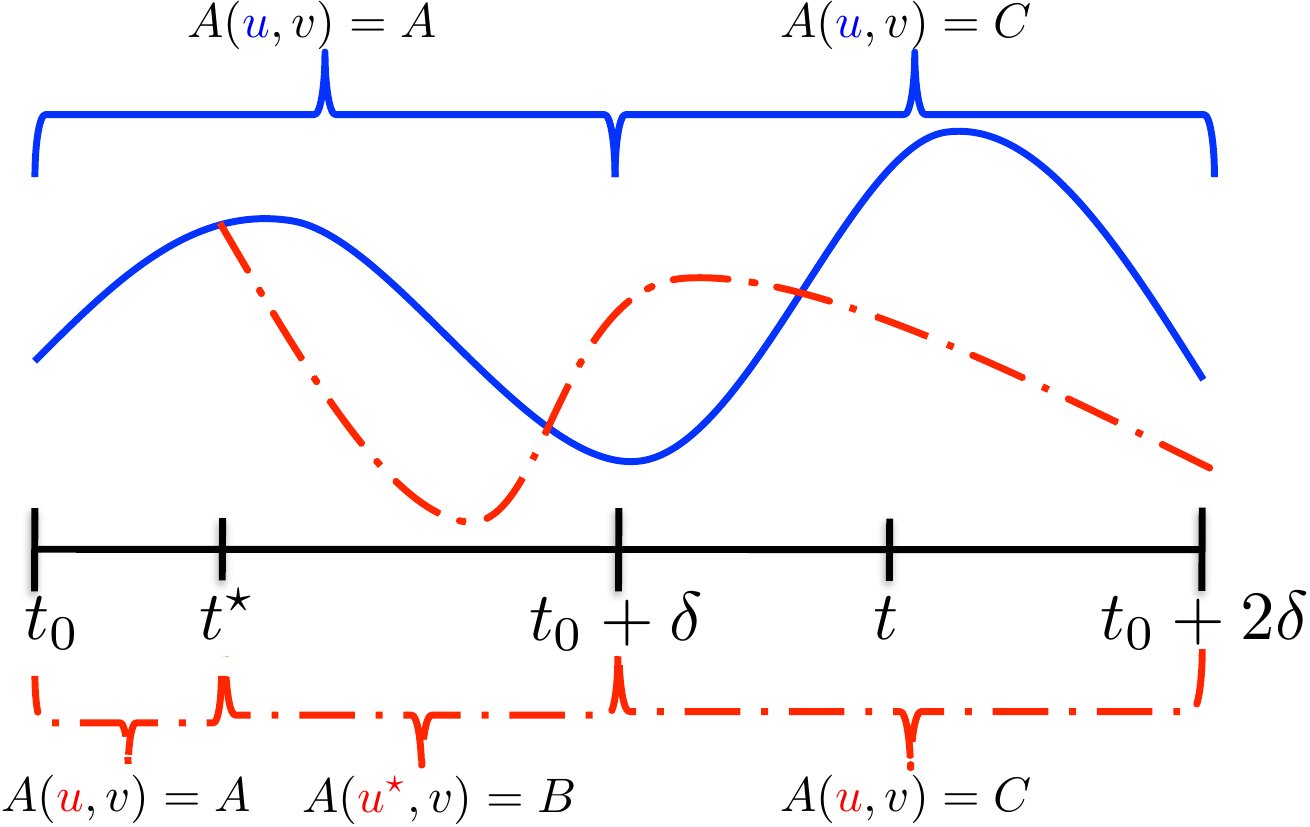}
\caption{A demonstration of the technique used in the proof. The blue solid trajectory corresponds to $u$ while the red dashed trajectory corresponds to $u^\star$.}
\label{fig::proofSketch}
\end{figure}
To show that the gain obtained over $[t_0,t_0+\delta]$ by the switch made by $u^\star$ is maintained over $[t_0+\delta,t_0+2\delta]$, we must prove the following inequality:
\begin{eqnarray}
&&\int_{t_0+\delta}^{t_0 + 2\delta} k(t)\cdot  \left[  L_1 -   L_2 \right] dt > 0,\label{eqn::ineqToProveProp}
\end{eqnarray}
where
\begin{eqnarray*}
L_1 & := &  \vnorm{(R(t-(t_0+\delta))-M )Q(t_0+\delta - t^\star)P(t^\star -t_0) x(t_0)}_2^2,\\
L_2 & := & \vnorm{ (R(t-(t_0+\delta))-M )P(t_0 + \delta -t_0) x(t_0)}_2^2.
\end{eqnarray*}
As before, it suffices to prove that the integrand $L_1-L_2$ is positive. Let us now expand both $L_1$ and $L_2$.
\begin{eqnarray*}
L_1 & = & x(t_0)^TP(\ts -t_0)Q(t_0+\delta-\ts)(R(t-(t_0+\delta))-M)(R(t-(t_0+\delta))-M)\\
&&Q(t_0+\delta-\ts)P(\ts -s)x(t_0)\\
& = & x(t_0)^T(P(\ts -t_0)Q(t_0+\delta-\ts)R(2(t-(t_0+\delta)))Q(t_0+\delta-\ts)P(\ts -t_0)-M)x(t_0).
\end{eqnarray*}
Similarly, $L_2 = x(t_0)^T(P(\delta)R(2(t-(t_0+\delta)))P(\delta)-M)x(t_0)$. We can then write
\begin{eqnarray*}
L_1 - L_2 &=& x(t_0)^T(P(\ts -t_0)Q(t_0+\delta-\ts)R(2(t-(t_0+\delta)))Q(t_0+\delta-\ts)P(\ts -t_0)\\ 
&&- P(\delta)R(2(t-(t_0+\delta)))P(\delta))x(t_0)\\
& := & x(t_0)^T(F_1-F_2)x(t_0).
\end{eqnarray*}
Before we perform a first-order Taylor expansion to the above terms, let us define the following quantities: $\tau_1 = t^\star -t_0$, $\tau_2 = (t_0+\delta) - t^\star$, and $\tau_3 = t - (t_0+\delta)$, where $t^\star \in [t_0,t_0+\delta]$ and $t \in [t_0+\delta, t_0+2\delta]$.
Using Proposition \ref{prop::bigO} in the Appendix, we can now expand $F_1$ and $F_2$ as follows:
\begin{eqnarray*}
F_1 & = & \left(I+\tau_1A+\bigO{\tau_1^2}\right)\left(I+\tau_2B+\bigO{\tau_2^2}\right)\left(I+2\tau_3C+\bigO{\tau_3^2}\right)\left(I+\tau_2B+\bigO{\tau_2^2}\right)\\
&&\left(I+\tau_1A+\bigO{\tau_1^2}\right) \\
& = & I + 2\tau_1A + 2\tau_2B + 2\tau_3C + \bigO{\delta^2} \\
F_2 & = & \left(I+\delta A + \bigO{\delta^2}\right)\left(I +2\tau_3C+\bigO{\tau_3^2}\right)\left(I+\delta A + \bigO{\delta^2}\right) \\
& = & I+2\delta A + 2\tau_3C + \bigO{\delta^2}.
\end{eqnarray*}
Hence, we have $F_1 - F_2 = 2\left(\left(t_0+\delta\right)-\ts\right)\left(B-A\right) +\bigO{\delta^2}$,
and thereby we obtain
\begin{eqnarray*}
L_1 - L_2 & = &  2\left(\left(t_0+\delta\right)-\ts\right) \sum_{r>s}(A_{sr} - B_{sr})\left(x_r(t_0)-x_s(t_0)\right)^2  + \bigO{\delta^2} \\
& = & 2\left(t_0+\delta -\ts\right)A_{ij}\left(x_j(t_0)-x_i(t_0)\right)^2 + \bigO{\delta^2}.
\end{eqnarray*}
If instead the matrix $B$ dictates breaking link $e_{ij}$ and unbreaking link $e_{kl}$, where $w_{ij} < w_{kl}$, the difference in the utilities would be $L_1 - L_2  = 2\left(t_0+\delta -\ts\right)(w_{kl}(t_0)-w_{ij}(t_0)) + \bigO{\delta^2}$. Hence, in both cases, for small enough $\delta$, we conclude that $L_1-L_2 >0$, which implies that (\ref{eqn::ineqToProveProp}) is satisfied, and the gain obtained by switching to system matrix $B$ at $t^\star \in [t_0,t_0+\delta]$ is maintained over $[t_0+\delta,t_0+2\delta]$. Note that the effect of switching to matrix $C$ is cancelled out in $F_1-F_2$, and hence $L_1-L_2$, since the strategy $u^\star$ is mimicking strategy $u$. Hence, by partitioning the interval $(t_0+2\delta,T]$ into small sub-intervals of length $\delta$ and repeating the above analysis, we conclude that the gain due to the switch at time $t^\star$ is preserved over the remaining time of the problem. 
\end{proof}

We can now derive the optimal strategy of the designer in the min--max problem. Recall the set $\mL_\ell(v) \subseteq \mE$ defined in \eqref{eqn::setL}. Let $\mL_{\ell,k}(v) \in \mE$ denote the $k$-th link of $\mL_\ell(v)$, $k \in \{1,\hdots, \ell \}$. Also, define $\mL^{-1}_{\ell,k}(v) \in \real$ as the value such that $\Phi(\mL_{\ell,k}(v),\mL^{-1}_{\ell,k}(v)) = \mL_{\ell,k}(v)$. We assume that $\mL^{-1}_{\ell,1}(v) \geq \hdots \geq \mL^{-1}_{\ell,\ell}(v)$. Further, define the sets $\mP(v) = \{(e_{ij},a_{ij}\nu_{ij}): e_{ij} \notin \mL_\ell(v)\} \subset \mE \times \real$ and $\overline{\mP}(v) = \{(e_{ij},\nu_{ij}): e_{ij} \notin \mL_\ell(v)\} \subset \mE \times \real$. We also define 
\begin{eqnarray*}
[v_{\mS}(b)]_{ij} & = & \left\{
  \begin{array}{l l}
    b, &  \text{$e_{ij} \in \mS$}  \\
    0, &  \text{$e_{ij} \notin  \mS$}
  \end{array} \right.
\end{eqnarray*}

\begin{theorem} \label{thm::minmaxMin}
In the min--max problem, and under Assumptions \ref{assume::costSwitch} and \ref{assume::commonInfo}, the optimal strategy of the designer is to run Algorithm \ref{algorithm} and set $v_{ij}^\star \in \{0,b\}$ if $\nu_{ij} = 0$. Further, it is optimal for the designer to modify $\ell$ links.
\end{theorem}
\begin{table}[!t]
 \centering
 \caption{
   Algorithm I: Computing the optimal strategy for the minimizer in the min--max problem.}\vspace*{-1em}
   \begin{tabular}{p{12cm}}
     \hline \\
0:\hspace*{1em} \textbf{input:} \text{a strategy $v$ with $\|v\|_1 = 0$}\vspace*{.2em}\\
1:\hspace*{1em} \textbf{for $i = \ell \downarrow  1$}\vspace*{.2em}\\
2: \hspace*{2em} \textbf{if} $\exists \mS \subseteq \Phi(\mP(0)), |\mS| = i$, $\mL_{\ell,i}(0) \notin \mL_\ell(v_\mS(b))$\vspace*{.1em}\\
3:\hspace*{4em} Set $v_{ij}^\star = b$, $\forall e_{ij} \in \mS \cup \Phi_{\ell-i}\left(\overline{\mP}(v_\mS(b))\right)$.\vspace*{.1em}\\
4:\hspace*{4em} \textbf{Exit} for loop.\vspace*{.1em}\\
5:\hspace*{3em}\textbf{end}\vspace*{.1em}\\
6:\hspace*{1em}\textbf{end}\vspace*{.1em}\\
7:\hspace*{1em}\textbf{if} $\|v\|_1 =0$ \vspace*{.1em}\\
8:\hspace*{2em} Set $v_{ij}^\star = b$ for all $e_{ij} \in \Phi_{\ell}\left(\overline{\mP}(0)\right)$.\vspace*{.1em}\\
9:\hspace*{1em}\textbf{end}\vspace*{.1em}\\ \\
  \hline
   \end{tabular}\label{tab:algo}\vspace{-0.5cm}
   \label{algorithm}
\end{table}

\begin{proof}
By Theorem \ref{thm::implement}, we deduce that $\|v^\star(t)\|_1 = b\ell$, because the designer would be at a disadvantage if he modifies fewer links than the adversary.

We first consider the designer's strategy over a fixed small interval $[t_0,t_0+\delta]$ over which both $u$ and $v$ are fixed. Using similar steps as those leading to (\ref{ineq::toProve}), and after applying a first-order Taylor expansion, we can write the designer's utility over $[t_0,t_0+\delta]$ as
\begin{equation} \label{eqn::utilSmall}
\int_{t_0}^{t_0+\delta} k(t)\cdot 2(t-t_0) \sum_{j>i} (a_{ij}+v_{ij})(1-u_{ij})(x_i(t_0) - x_j(t_0))^2 dt+ \bigO{\delta^2}.
\end{equation}
According to Theorem \ref{thm::implement}, and in the absence of the designer, it is optimal for the adversary to break the links in $\mL_\ell(0)$. Therefore, the designer must attempt to modify the ranking of the links such that the links (or a subset of them) in $\mL_\ell(0)$ are not in $\mL_\ell(v^\star)$. In essence, this is what Algorithm \ref{algorithm} attempts to achieve. Being of the lowest negative value, and hence the link both the adversary and the designer are interested in, let us explore how the designer can push $\mL_{\ell,\ell}(0)$ higher in the ranking of the link values. The designer can achieve this if under some strategy $v\in \mV$, the value $\mL^{-1}_{\ell,\ell}(0)$ is no longer among the lowest $\ell$ negative values; in other words, the designer can alter the ranking if there is a set $\mS \subset \mP(0)$, $|\mS|=\ell$, such that when he sets $v_{ij} = b$ for all links in $\mS$, there will be $\ell$ values that are smaller than $\mL^{-1}_{\ell,\ell}(0)$ (steps $2$ and $3$ in Algorithm \ref{algorithm}). The adversary will then break the links in $\mS$ and will spare the link corresponding to $\mL^{-1}_{\ell,\ell}(0)$ as required. To see why this is optimal, consider the following two cases, covering the types of links that can be in $\mS$.
\smallbreak
\noindent \textbf{Case 1:} If a link in $\mS$ is also in $\mL_\ell(0)$, then this is optimal due to the fact that the adversary will disconnect that link since it is in $\mL_\ell(0)$. Hence, if the designer can utilize this link to modify the ranking and protect a link whose associated value is more negative ($\mL_{\ell,\ell}(0)$ in this case), then this can only improve his utility. The same reasoning applies if more than one of the links in $\mS$ are also in $\mL_\ell(0)$.
\smallbreak
\noindent \textbf{Case 2:} If none of the links in $\mS$ is in $\mL_\ell(0)$, then necessarily some of the links in $\mL_\ell(0)$ will also be protected along with the link corresponding to $\mL^{-1}_{\ell,\ell}(0)$. This is because $|\mS| = \ell$, and the adversary can break at most $\ell$ links. Hence, this scenario is more favorable to the designer than the previous one and can therefore only improve his utility.

If such an $\mS$ exists, then the designer would have exhausted all possible moves, since $|\mS| = \ell$, and the algorithm terminates (step $4$ of the algorithm). Otherwise, if no such set exists in $\mP(0)$, then the designer should try to protect the next most negative link whose value is precisely $\mL^{-1}_{\ell,\ell-1}(0)$ by finding a set $\mS$ of size $\ell-1$. Since $\mL^{-1}_{\ell,\ell-1}(0) \geq \mL^{-1}_{\ell,\ell}(0)$, the link corresponding to $\mL^{-1}_{\ell,\ell}(0)$ along with $\mS$ will constitute the set of $\ell$ links that the adversary will break. Then, the designer should set $v_{ij}=b$ for all the links in $\mS$, and for the remaining action the designer should select the link with the most negative $\nu_{ij}$ that is \emph{not} in $\mL_\ell(v_\mS(b))$; this is precisely the set $\Phi_{1}\left(\overline{\mP}(v_\mS(b))\right)$ (step $3$ of the algorithm). The reason behind searching in $\overline{\mP}(v_\mS(b))$ and not in $\mP(v_\mS(b))$ after finding $\mS$ is that the $a_{ij}$'s only affect the utility of the designer when he attempts to alter the ranking. 

This procedure then repeats until the designer has tried to protect all the links in $\mL_\ell(0)$. If the designer fails in protecting \emph{all} the links in $\mL_\ell(0)$, then we must have $\|v\|_1 = 0$, i.e., the input strategy was not altered. Then, the optimal strategy is to set $v_{ij} =b$ for the links with most negative $\nu_{ij}$'s in $\overline{\mP}(0)$ (steps $7$ and $8$ in Algorithm \ref{algorithm}).

The final step of the proof is to show that applying Algorithm \ref{algorithm} over $[0,T]$ is optimal for the designer. To this end, it suffices to show that modifying links with lower $\nu_{ij}$ values is more beneficial to the designer, as Algorithm \ref{algorithm} attempts to protect these links. Given the links $e_{ij},e_{kl} \in \mE$, assume that $\nu_{ij} < \nu_{kl}$. Consider the two system matrices $A$ and $B$, and let $v_{ij}=0$, $v_{kl}=b$ and $v^\star_{ij}=b$, $v^\star_{kl}=0$. Assume that a strategy $v$ dictates applying matrix $A$ over $[t_0,t_0+\delta]$ and applying the matrix $C$ over $[t_0+\delta,t_0+2\delta]$. Also, assume that according to $v^\star$, the designer applies $A$ over $[t_0+\delta,t^\star)$, $B$ over $(t^\star,t_0+\delta]$, and $C$ over $[t_0+\delta,t_0+2\delta]$. Following the steps presented in step 2 of the proof of Theorem \ref{thm::implement}, we conclude that, for $\delta$ small enough, the quantity $-2\left(t_0+\delta -\ts\right)b(\nu_{kl} - \nu_{ij}) + \bigO{\delta^2}$ is negative. It then follows that the gain obtained by switching to system matrix $B$ at $t^\star \in [t_0,t_0+\delta]$ is maintained over $[t_0+\delta,t_0+2\delta]$. Hence, by partitioning the interval $(t_0+2\delta,T]$ into small sub-intervals of length $\delta$ and repeating the above analysis, we conclude that Algorithm \ref{algorithm} is optimal over $[0,T]$. 
\end{proof}

\subsection{The Max--Min Problem}
The following theorem specifies the optimal strategies of the adversary and the designer in the max--min problem. Let $\mathcal{F}_\ell(u) = \Phi_\ell(\{(e_{ij},\nu_{ij}): e_{ij} \in \mathcal{E}(u)\}) \subset \mE(u)$, where we recall that $\mE(u) = \mE \setminus \{e_{ij} \in \mE : u_{ij}(t)=1\}$, for some $u \in \mU$. {If $m < 2\ell$, the sets $\mE(u),\mathcal{F}_\ell(u)$  could contain fewer than $\ell$ links. For simplicity, we assume that $m \geq \ell$ in the following proof, which guarantees that $|\mathcal{F}_\ell(u)|=\ell$. However, the result of the theorem applies regardless of this assumption, and the modification of the proof is straightforward.}
\begin{theorem} \label{thm::maxminMin}
Under Assumptions \ref{assume::costSwitch} and \ref{assume::commonInfo}, and for a fixed strategy $u$ of the adversary, the optimal strategy of the network designer in the max--min problem is given by
\begin{eqnarray*}
v^\star_{ij}(u)  & = &   \left\{
  \begin{array}{l l}
    b, &  \text{$e_{ij} \in \mathcal{F}_{\ell}(u)$}  \\
    0, &  \text{$e_{ij} \notin \mathcal{F}_{\ell}(u)$}  
      \end{array} \right. 
  \end{eqnarray*}
If the designer has an optimal strategy of modifying fewer than $\ell$ links, then either $\mathcal{G}$ has a cut of size less than $\ell$ or the nodes have reached consensus by time $t$. In either of these cases, breaking $\ell$ links is also optimal.
\end{theorem}
\begin{proof}
The proof follows the same two steps used to prove Theorem \ref{thm::implement}. For a fixed strategy of the adversary $u$, we will show that it is optimal for the minimizer to rank the links based on their $\nu_{ij}$ values. Under Assumption \ref{assume::costSwitch}, the function $x$ becomes piecewise continuous. Hence, the function $\nu_{ij}$, for all $e_{ij}\in \mE(u)$, is also piecewise continuous and its value cannot change abruptly over a finite interval. As a result, we can regard the system as a time-invariant one over a small interval $[t_0,t_0+\delta] \subset [0,T]$, where $0 < \delta \leq \tau$, and $\tau$ was defined in (\ref{eqn::dwellTime}). 

Let $v$ be an arbitrary strategy of the designer with $\|v\|_1 < b\ell$. Over a small interval, $v$ and $v^\star$ induce certain system matrices. Let the system matrix corresponding to $v$ over $[t_0,t_0+\delta]$ be $A(u,v)=A$. Because the control strategies of both players are time-invariant over this interval, the state trajectory is given by (\ref{eqn::stateFixed}). We want to show that switching from strategy $v$ to strategy $v^\star$ at some time $t^\star \in [t_0,t_0+\delta]$ can improve the utility of the designer. To this end, we assume that the matrix induced by $v^\star$ over $[t_0,t^\star)$ is $A$, while the system matrix corresponding to $v^\star$ over $[t^\star,t_0+\delta]$ is $B$. Assume that $e_{ij} \in \mE(u)$, i.e., $u_{ij} = 0$. Over $[t^\star,t_0+\delta]$, the strategies $v$ and $v^\star$ are identical except at link $e_{ij}$, where $v_{ij} = 0$ and $v_{ij}^\star = b$, i.e., $\|v\|_1 < \|v^\star\|_1$ over this sub-interval. It follows that:
\begin{equation}
B_{ij}=a_{ij}+b > A_{ij} = a_{ij}, \quad A_{kl} = B_{kl}, \quad \forall  e_{kl} \neq e_{ij}.  \label{eqn::AvsB2}
\end{equation}
Following similar steps to those in the proof of Theorem \ref{thm::implement}, we conclude that it suffices to prove 
\begin{equation*} 
h(t,x(t_0)) = x(t_0)^T\Lambda(t,t^\star)x(t_0)<0, \text{ for } t> t^\star,
\end{equation*}
where $\Lambda(t,t^\star)$ was defined in the proof of Theorem \ref{thm::implement}. For sufficiently small $\delta$, we can arrive at the expansion in (\ref{eqn::matApprox}). Using (\ref{eqn::AvsB2}) and properties of Laplacian matrices, we can then write
\begin{eqnarray}
h(t,x(t_0)) & = & 2(t-t^\star)  \sum_{r>s}(A_{sr} - B_{sr})\left(x_r(t_0)-x_s(t_0)\right)^2  + \bigO{\delta^2}\nonumber \\
& = & -2(t-t^\star)b\left(x_j(t_0)-x_i(t_0) \right)^2 + \bigO{\delta^2}. \label{eqn::expandH}
\end{eqnarray}
For small enough $\delta$, the higher order terms are dominated by the first term. Hence, if there is a link $e_{ij}$ such that $x_i(t_0) \neq x_j(t_0)$, there exists $t^\star$ such that $h(t,x(t_0))<0$ for $t \in \left(t^\star,t_0+\delta\right]$. Since $t_0$ was arbitrary, we conclude that the optimal strategy must satisfy $\|v^\star(t)\|_1 = b\ell$ for all $t$, given that each of the $\ell$ links connects two nodes having different values. 

If no link such that $x_i(t_0) \neq x_j(t_0)$ exists at a given time $t_0$, the designer does not need to break additional links, although breaking more links does not affect optimality because $h(t,x(t_0))=0$ in such a case. There are two cases where the designer cannot find a link to make $h(t,x(t_0))<0$, and they were presented in the proof of Theorem \ref{thm::implement} in the case of the adversary. However, unlike the case of the adversary, Case (i) presents a winning strategy for the designer as the nodes are in agreement. Case (ii) is not necessarily a winning or a losing strategy for the designer.

Next, we need to show that the designer will modify the $\ell$ links in $\mE(u)$ with the lowest $\nu_{ij}$ values. Let us again restrict our attention to the interval $[t_0,t_0+\delta]$ where the designer applies strategy $v$. Assume (to the contrary) that the links the designer modifies over this interval are not the ones with the lowest $\nu_{ij}$ values. In particular, assume that the designer chooses to modify link $e_{kl} \in \mE(u)$, while there is a link $e_{ij} \in \mE(u)$ such that $\nu_{ij} < \nu_{kl}$. Assume that the designer switches at time $t^\star \in [t_0,t_0+\delta]$ to strategy $v^\star$ by modifying link $e_{ij}$ instead of link $e_{kl}$. Then, (\ref{eqn::expandH}) becomes
\[
h(t,x(t_0))=-2(t-t^*)b\left(\nu_{kl}(t_0)-\nu_{ij}(t_0)  \right) + \bigO{\delta^2}.
\]
Hence, by following the same arguments as above, we can conclude that modifying $e_{kl}$ is not optimal. 

The second step of the proof is to show that switching to strategy $v^\star$ guarantees an improved utility for the designer regardless of how the original trajectory corresponding to $v$ changes beyond time $t_0+\delta$. To this end, we will assume that from time $t_0+\delta$ onward, strategy $v^\star$ will mimic strategy $v$. Assume that strategy $v$ switches from matrix $A$ to matrix $C$ over the interval $[t_0+\delta,t_0+2\delta]$. Hence, strategy $v^\star$ will also switch from the system matrix $B$ to matrix $C$. However, the trajectories corresponding to $v$ and $v^\star$ will have different initial conditions at time $t_0+\delta$, due to the switch that strategy $v^\star$ made at time $t^\star$. Recall that according to $A$, we have $\|v\|_1 < b\ell$ and $v_{ij} = 0$. Here, the system matrix $B$ can differ from the matrix $A$ in two ways: either (i) $B$ dictates modifying one additional link compared to $A$, or (ii) $B$ dictates modifying link $e_{ij}$ instead of link $e_{kl}$ where $\nu_{ij} < \nu_{kl}$. Consider the behavior of the system over the interval $[t_0+\delta, t_0 + 2\delta]$ where we can assume that the system is time-invariant. To show that the gain obtained over $[t_0,t_0+\delta]$ by the switch made by $v^\star$ is maintained over $[t_0+\delta,t_0+2\delta]$, it suffices to prove that the integrant $L_1-L_2$ is negative, where $L_1$ and $L_2$ were defined in the proof of Theorem \ref{thm::implement}. For Case (i), by following the steps presented in the proof of Theorem \ref{thm::implement}, we can write
\begin{eqnarray*}
L_1 - L_2 & = &  -2\left(t_0+\delta -\ts\right)b\left(x_j(t_0)-x_i(t_0)\right)^2 + \bigO{\delta^2}.
\end{eqnarray*}
For Case (ii), the difference in utilities would be
\begin{eqnarray*}
L_1 - L_2  = 2\left(t_0+\delta -\ts\right)(w_{kl}(t_0)-w_{ij}(t_0)) + \bigO{\delta^2}.
\end{eqnarray*}
Hence, for small enough $\delta$, we conclude that $L_1-L_2 <0$. By partitioning the interval $(t_0+2\delta,T]$ into small sub-intervals of length $\delta$ and repeating the above analysis, we conclude that the gain due to the switch at time $t^\star$ is preserved over the remaining time of the problem. This concludes the proof.
\end{proof}

Next, we present the optimal strategy of the adversary. To this end, define the set 
$$
\mathcal{D}_\ell = \Phi_\ell(\{(e_{ij,}a_{ij}\nu_{ij}) : e_{ij}\in \mathcal{E}\} \cup \{(e_{ij},(a_{ij}+b)\nu_{ij}) : e_{ij}\in \mathcal{E} \}).
$$ 
\begin{theorem} \label{thm::maxminMax}
In the max--min problem, and under Assumptions \ref{assume::costSwitch} and \ref{assume::commonInfo}, the optimal strategy of the adversary is given by
\begin{eqnarray*}
u^\star_{ij}(t) & = & \left\{
  \begin{array}{l l}
    1, &  \text{$e_{ij} \in \mD_{\ell}$}  \\
    0, &  \text{$e_{ij} \notin  \mD_{\ell}$}
      \end{array} \right.
\end{eqnarray*}
Further, it is optimal for the adversary to break $\ell$ links.
\end{theorem}
\begin{proof}
By Theorem \ref{thm::maxminMin}, we deduce that $\|u^\star(t)\|_1 = \ell$, because the adversary would be at a disadvantage if he breaks fewer links than the designer. We first consider the adversary's strategy over a fixed small interval $[t_0,t_0+\delta]$ over which both $u$ and $v$ are fixed. Using a first-order Taylor expansion, the adversary's utility over $[t_0,t_0+\delta]$ is given by (\ref{eqn::utilSmall}).

In this problem, the adversary has the first-mover-advantage and needs to dispose of the links that can reduce his utility. The adversary knows that, according to $v^\star(u)$, the designer attempts to make the $\nu_{ij}$'s smaller by adding $b$ to the corresponding edge weights. However, we cannot rule out the possibility that $(a_{lk}+b)\nu_{lk} > a_{ij}\nu_{ij}$, for some links $e_{kl}$ and $e_{ij}$. Hence, the adversary is not only interested in finding the smallest negative $(a_{ij}+b)\nu_{ij}$'s, but also needs to consider the $a_{ij}\nu_{ij}$'s themselves. It follows that the adversary needs to find the terms that can become very small (negative) and set $u_{ij}=1$ to the corresponding links. But those links are exactly the ones included in $\mD_\ell$. Formally, we can write
\begin{eqnarray*}
 -\sum_{\substack{j>i\\ e_{ij}\in \mD_\ell}}  (a_{ij}+v_{ij})\nu_{ij} \leq -\sum_{\substack{j>i\\ e_{ij}\notin \mD_\ell}}  (a_{ij}+v_{ij})\nu_{ij},
\end{eqnarray*}
This confirms that, over the interval $[t_0,t_0+\delta]$, $u^\star$ is as claimed. 

The final step of the proof is to show that switching from a strategy $u$ to strategy $u^\star$ guarantees an improved utility for the designer over $[0,T]$. To this end, it suffices to show that modifying links with lower $w_{ij}$ values is more beneficial to the adversary. For the links $e_{ij},e_{kl} \in \mE$, assume that $w_{ij} < w_{kl}$. Consider the two system matrices $A$ and $B$, and let $u_{ij}=0$, $u_{kl}=1$ and $u^\star_{ij}=1$, $u^\star_{kl}=0$. Assume that the strategy $u$ dictates applying matrix $A$ over $[t_0,t_0+\delta]$ and applying the matrix $C$ over $[t_0+\delta,t_0+2\delta]$. On the other hand, we assume that according to $u^\star$, the adversary applies $A$ over $[t_0+\delta,t^\star)$, $B$ over $(t^\star,t_0+\delta]$, and $C$ over $[t_0+\delta,t_0+2\delta]$. Following the steps presented in step 2 of the proof of Theorem \ref{thm::implement}, we conclude that, for $\delta$ small enough, the quantity $2\left(t_0+\delta -\ts\right)(w_{kl} - w_{ij}) + \bigO{\delta^2}$ is positive, which implies that the gain obtained by switching to system matrix $B$ at $t^\star \in [t_0,t_0+\delta]$ is maintained over $[t_0+\delta,t_0+2\delta]$. Hence, by partitioning the interval $(t_0+2\delta,T]$ into small sub-intervals of length $\delta$ and repeating the above analysis, we conclude that $u^\star$ is optimal over $[0,T]$. 
\end{proof}

\begin{remark}
(Potential-Theoretic Analogy)
When the graph is viewed as an electrical network, $a_{ij}+v_{ij}$ can be viewed as the conductance of link $e_{ij} \in \mE$, and $x_i-x_j$ as the potential difference across the link. Therefore, according to Theorems \ref{thm::minmaxMin} and \ref{thm::maxminMin}, the optimal strategy of the designer in both problems involves finding the links with the highest potential difference (or the lowest $\nu_{ij}$'s) and increasing the conductance of those links by setting $v_{ij} = b$. This leads to increasing the power dissipation across those links, which translates to increasing the information flow across the network and results in faster convergence. The optimal strategy of the adversary should therefore involve breaking the links with the highest power dissipation. But power dissipation is given by $(a_{ij}+v_{ij})(x_i-x_j)^2$, and this is exactly what the adversary targets according to Theorems \ref{thm::implement} and \ref{thm::maxminMax}.
\end{remark}

\subsection{From Potential Theory to the Maximum Principle}
In this section, we show that the strategies derived in the above theorems satisfy the first-order necessary conditions for optimality given by the maximum principle (MP). We will address here the min--max problem; a theorem similar to the one presented below can be obtained also for the max--min problem. In \cite{KhanaferTouriBasarCDC13}, we showed that the optimal strategies provided by the MP for the min--max problem are the same as those derived in Theorems \ref{thm::implement} and \ref{thm::minmaxMin}, with the ranking of the links performed after replacing the quantity $\nu_{ij}$ with the quantity $(p_j-p_i)(x_i - x_j)$, where $p$ is the costate vector. The next theorem states that the potential-theoretic strategies satisfy the MP if the controllers do note switch infinitely many times over $[0,T]$.
\begin{theorem} \label{thm::anotherProof}
Under Assumptions \ref{assume::costSwitch} and \ref{assume::commonInfo}, the optimal strategies in Theorems \ref{thm::implement} and \ref{thm::minmaxMin} satisfy the canonical equations of the MP.
\end{theorem}
\begin{proof}
See the Appendix.
\end{proof}
\subsection{Complexity of the Optimal Strategies}
We next study the complexity of the optimal strategies. We first start with the max--min problem. Assuming, as in Remark \ref{rem::complex}, that the players switch their strategies a total of $K$ times over $[0,T]$, we conclude that the worst-case complexity of the strategy of either player is $\mathcal{O}(K\cdot m\log m)$ as their strategies involve merely the ranking of sets of size at most $2m$. As for the min--max problem, the complexity of the adversary's strategy is $\mathcal{O}(K\cdot m\log m)$. The main bottleneck in the strategy of the designer is step $2$ in Algorithm \ref{algorithm}. The size of the set $\mP(0)$ is at most $m-\ell$; thus, the worst-case complexity for the designer is $K \cdot \sum_{i=1}^{m-\ell}{m-\ell \choose i} \approx K\cdot \sum_{i=1}^\ell(m-\ell)^i$. By comparison with (\ref{bruteForce}), we conclude that the optimal strategies achieve vast complexity reductions.

\subsection{An Illustrative Example} \label{sec::example}
The goal of this example is twofold: (i) to show how the players execute their strategies; and (ii) to serve as a counter example showing that an SPE may not exist and to provide some guidelines as to when one would exist. We will study the interaction between the designer and the adversary for the case when $T=\tau$, and $\tau$ is very small. By Assumption \ref{assume::costSwitch}, we conclude that the players cannot change the actions they choose at time $t=0$. Assume that $\mG$ is a complete graph with three nodes with the following weights:
\begin{eqnarray*}
A(0,0) = \left[\begin{array}{ccc}-4 & 3 & 1 \\3 & -5 & 2 \\1 & 2 & -3\end{array}\right].
\end{eqnarray*}
Define $e_1 = (1,2)$, $e_2 = (2,3)$, $e_3 = (1,3)$. Let $\nu_{12}=-1$, $\nu_{23}=-2$, and $\nu_{13}=-5$. Let $x(0)=[1,2,3]^T$ and $\ell = 1$. Consider the following two cases:

\noindent \textbf{Case 1:} \emph{($b =1$)}  Let us first consider the max--min problem. We have $\mD_{1} =\Phi_1( \{(e_1,-3),(e_1,-4),\\(e_2,-4),(e_3,-5),(e_2,-6),(e_3,-10) \}) = \{e_3 \}$. Hence, according to Theorem \ref{thm::maxminMax}, the adversary breaks $e_3$, and we have that $\mE(u^\star) = \mE\setminus e_3$. We also have $\mF_1(u^\star) =\{e_2\}$, which means that $v^\star = [0,1,0]^T$ and $u^\star = [0,0,1]^T$. Hence, using (\ref{eqn::utilSmall}), we can write
\begin{eqnarray*} 
\underline{V}&=&\int_{0}^{T} k(t)\cdot 2t [3(x_1(0) - x_2(0))^2+ 3(x_2(0) - x_3(0))^2] dt+ \bigO{\delta^2} \\
&=& \int_{0}^{T} k(t)\cdot 12t dt+ \bigO{\delta^2}.
\end{eqnarray*}

For the min--max problem, Algorithm \ref{algorithm} uses the following sets $\mL_1(0) = \{e_3 \}$ and $\mP(0) = \{(e_1,-3),(e_2,-4) \}$. Let $\mS = \{e_2\}$, and note that $\mS \in \Phi(\mP(0)) = \{e_1, e_2 \}$. We then have $v_S(1) = [0,1,0]^T$ and $\mL_1(v_S(1)) = \{e_2 \}$. Note that $\mL_1(0) \notin \mL_1(v_\mS(1))$. Hence, the condition in step $2$ of the algorithm is satisfied with this choice of $\mS$, and we have $v^\star = v_\mS(1)$. Then, Theorem \ref{thm::minmaxMin} says that the designer will increase the weight of $e_2$, and Theorem \ref{thm::implement} says that the adversary will break the same link, i.e., $v^\star = [0,1,0]^T$ and $u^\star = [0,1,0]^T$. We thus have
\begin{equation*} 
\overline{V}=\int_{0}^{T} k(t)\cdot 14t dt+ \bigO{\delta^2}.
\end{equation*}
We conclude that in this case $\overline{V} > \underline{V}$, and an SPE does not exist.

\noindent \textbf{Case 2:} \emph{($b =0.4$)} By repeating the above steps, we conclude that in the max--min problem we have $v^\star = [0,0.4,0]^T$ and $u^\star = [0,0,1]^T$, and we can write
\begin{equation*} 
\underline{V}=\int_{0}^{T} k(t)\cdot 10.8t dt+ \bigO{\delta^2}.
\end{equation*}

For the min--max problem, one cannot find a set $\mS$ satisfying the conditions of step 2 in Algorithm \ref{algorithm}. To execute step 8 of the algorithm, note that $\mL_1(0) = \{e_3 \}$, and hence $\Phi_1(\overline{\mP}(0)) = \{e_2\}$. We therefore have $v^\star = [0,0.4,0]^T$ and $u^\star = [0,0,1]^T$, and hence
\begin{equation*} 
\overline{V}=\int_{0}^{T} k(t)\cdot 10.8t dt+ \bigO{\delta^2}.
\end{equation*}
In this case, the pair of inequalities (\ref{eqn::SPE}) are satisfied and an SPE exists. The main difference between the two cases was that the designer was able to find a set $\mS$ that allows him to alter the ranking and deceive the adversary when $b=1$. This made the adversary break $e_3$ in the max--min problem and break $e_2$ in the min--max problem which led to having $\overline{V} \neq \underline{V}$. When such a set does not exit, the strategy of the adversary is unchanged in both problems, and hence the upper and lower values would agree. Hence, for an SPE to exist, one needs a behavior similar to Case 2 to occur throughout the problem horizon $[0,T]$. This of course depends on the value of $b$ and the weights $a_{ij}$. Section \ref{suffCond} explores the question of existence of an SPE further.

\section{A Sufficient Condition for the Existence of an SPE} \label{suffCond}
Thus far, we have solved the min--max and max--min problems separately and showed that the derived optimal strategies achieve the upper and lower values. Hence, to prove the existence of an SPE, it remains to verify whether the pair of inequalities (\ref{eqn::SPE}) can be satisfied under some assumptions, even though the action sets of the players are non-rectangular in the max--min problem. Besides the issue of non-rectangular action sets, the main reason that the upper and lower values are different is mainly due to the ability of the minimizer to \emph{deceive} the maximizer by altering the ranking of the most negative values. If we remove this ability from the network designer, we should expect that an SPE would exist. The following theorem makes this argument formal. Define $\gamma :=\frac{4\vnorm{x_0}_\infty^2}{\epsilon^2}$, $\epsilon>0$. We assume that $\epsilon$ is chosen to guarantee $\gamma > 1$.
\begin{theorem}
Given $\epsilon >0$, assume that $T$ is small enough such that (\ref{eqn::upperLowerBoudns}) in the Appendix holds. Then, under Assumptions \ref{assume::costSwitch} and \ref{assume::commonInfo}, a sufficient condition for the existence of an SPE for the underlying zero-sum game between the designer and the adversary is to select $b$ such that
\begin{equation} \label{eqn::sufCond}
0 \leq b \leq \min_{e_{ij}, e_{kl} \in \mE} \left| \gamma a_{ij} - a_{kl} \right|,
\end{equation}
given that $a_{ij} \neq a_{kl}$ and $a_{ij} > \gamma a_{kl}$ whenever $a_{ij} > a_{kl}$, for all $e_{ij},e_{kl}\in \mE$.
\end{theorem}
\begin{proof}
It suffices to show that $\mL_\ell(v^\star) = \mL_\ell(0) = \mD_\ell$ as this would imply that the adversary would break the same links whether he acts first or second, and as a result the strategy of the minimizer in both problems will be the same. This will guarantee that (\ref{eqn::SPE}) is satisfied. This would occur if the minimizer cannot protect any of the links in $\mL_\ell(0)$. In other words, this will happen if the minimizer cannot satisfy the condition in step $2$ of Algorithm \ref{algorithm} for any $i\in \{1,\hdots,\ell\}$. A sufficient condition for $\mL_\ell(v^\star) = \mL_\ell(0) = \mD_\ell$ to hold is to require
\begin{eqnarray*}
\min_{e_{ij}\in \Phi(\mP(0))}(a_{ij}+b)\nu_{ij} & > & \max_{e_{ij}\in \mL_\ell(0)}a_{ij}\nu_{ij}.
\end{eqnarray*}
This implies that no matter how the designer changes the weights of the links in $\Phi(\mP(0))$, he cannot make those links more negative than the links in $\mL_\ell(0)$. To satisfy this inequality, we will establish that whenever $a_{ij}\nu_{ij} > a_{kl}\nu_{kl}$, we must have $(a_{ij}+b)\nu_{ij} > a_{kl}\nu_{kl}$, for all $e_{ij},e_{kl} \in \mE$. We can then re-write the condition on $b$ as
\begin{eqnarray} \label{eqn::condOnb}
b \leq \frac{a_{ij}\nu_{ij} - a_{kl}\nu_{kl}}{-\nu_{ij}} = a_{kl} \frac{|\nu_{kl}|}{|\nu_{ij}|} - a_{ij}, \quad \forall e_{ij},e_{kl}\in \mE
\end{eqnarray}
Consider the following two cases. If $\nu_{kl}\geq \nu_{ij}$, then we must have $a_{kl}>a_{ij}$. Then, by assumption we have that $a_{kl} > \gamma a_{ij}$. By Lemma \ref{lemma::upperLowerBounds} in the Appendix, we can write
\begin{equation} \label{condB}
a_{kl} \frac{|\nu_{kl}|}{|\nu_{ij}|} - a_{ij}  \geq \frac{1}{\gamma} a_{kl} - a_{ij}>0.
\end{equation}
Next, consider the case when $\nu_{ij}>\nu_{kl}$. In this case, $a_{ij}$ can be larger or smaller than $a_{kl}$. However, if $a_{ij}>a_{kl}$, and recalling that $a_{ij}\nu_{ij} > a_{kl}\nu_{kl}$, then
\begin{equation*}
\gamma a_{kl} < a_{ij} < a_{kl}\frac{|\nu_{kl}|}{|\nu_{ij}|}\leq \gamma a_{kl},
\end{equation*}
which is a contradiction. The case $a_{kl}=a_{ij}$ is excluded by assumption. Hence, in this case, we must have $a_{ij}<a_{kl}$, and the inequality in (\ref{condB}) applies. Thus, by choosing $b$ as in (\ref{eqn::sufCond}), we obtain the condition we are seeking. Note that we do not need to consider the case when $a_{ij}\nu_{ij} = a_{kl}\nu_{kl}$ since the players will be indifferent as to which link to choose.
\end{proof}

\begin{remark}
The condition derived in the above theorem requires the network to be ``sufficiently diverse" in the sense that the weights of the links have to be not only different from each other, but also a factor $\gamma$ apart. This is due to the fact that we were seeking uniform bounds on the $\nu_{ij}$'s, for all $e_{ij} \in \mE$. If we allow $b$ to vary with time, then one can find less restrictive conditions to ensure the existence of an SPE. However, this would require (\ref{eqn::condOnb}) to be verified at each time instant. Further, the bound derived in (\ref{condB}) is loose, because it was obtained by bounding $|v_{kl}|$ and $|v_{ij}|$ independently, for $e_{ij},e_{kl} \in \mE$. Tighter bounds could be given by studying the dynamics of $|\nu_{kl}|/|\nu_{ij}|$. However, studying the time derivative of this ratio is not tractable.
\end{remark}

\begin{remark}
This result highlights the fact that, in general, Stackelberg games are more natural to study security problems than zero-sum games. In fact, the leader-follower formulation fits many real-world security scenarios; see \cite{korzhyk2011stackelberg} and the references therein. However, the sufficient condition we derive here is a step in the right direction for establishing the existence of an SPE for the zero-sum game between the designer and the adversary. We are currently investigating whether this condition is also necessary.

\end{remark}

\section{Conclusion} \label{Conclusion}
In this paper, we have studied the impact of an adversarial attack on a network of agents performing consensus averaging. The adversary's objective is to slow down the convergence of the computation at the nodes to the global average. We introduced a network designer whose objective is to assist the nodes reach consensus by countering the attack of the adversary. The adversary and the network designer are capable of targeting links. We have formulated and solved two problems that capture the competition between the two players. We considered practical models for the players by constraining their actions along the problem horizon. The derived strategies were shown to exhibit a low worst-case complexity. When Zeno behavior is excluded, we showed that the optimal strategies admit a potential-theoretic analogy. Finally, we showed that when the link weights are sufficiently diverse, an SPE exists for the zero-sum game between the designer and the adversary.

Future work will focus on removing Assumption \ref{assume::costSwitch} and showing that Zeno behavior can be ruled out in optimality. Formulating the problems in discrete-time is also of interest. Another interesting line of research is to derive the optimal strategies when the knowledge of the players about the state and the topology of the network is restricted. When applying necessary conditions for optimality, e.g., the MP, to the min--max or the max--min problem, one must first prove the existence of optimal controllers. Such results can be viewed as existence results for equilibria in the general framework of Stackelberg games. This is another avenue for future research.

\section*{Acknowledgment}
The authors would like to thank Prof. Behrouz Touri for valuable comments during the development of this work. We are also very grateful for the constructive comments made by the Associate Editor and the reviewers, which helped improve the original manuscript.

\bibliographystyle{IEEEtran}
\bibliography{references}

\appendix

\subsection{Proof of Theorem \ref{thm::anotherProof}  }
For a fixed strategy $v$ of the designer, it was shown in \cite{KhanaferTouriBasarCDC13} that the adversary's strategy derived using the MP requires finding the lowest $f_{ij} = (a_{ij}+v_{ij})(p_i-p_j)(x_j-x_i)$ values, for all $e_{ij} \in \mE$. However, Theorem \ref{thm::implement} requires finding the lowest $w_{ij}$'s. The designer's strategy relies on finding the lowest $(p_i-p_j)(x_j-x_i)$ values according to the MP, and it requires finding the lowest $\nu_{ij}$'s according to Theorem \ref{thm::minmaxMin}. In order to prove the theorem, and since $w_{ij} = (a_{ij}+v_{ij})\nu_{ij}$, $a_{ij}+v_{ij} \geq 0$, it is sufficient to show that $w_{ij} \leq w_{kl}$ implies that $f_{ij} \leq f_{kl}$, for all $e_{ij}$, $e_{kl} \in \mE$. Without loss of generality, we will assume that $v_{ij}=v_{kl}=0$. The Hamiltonian associated with the min--max problem is:
\[
H(x,p,u,v) = \frac{1}{2}k(t)\vnorm{x(t)-\bar{x}}_2^2 + p(t)^TA(u(t),v(t))x(t),
\]
where $p$ is the costate vector, whose existence is guaranteed by the MP because an optimal solution for the min--max exists. The first-order necessary conditions for optimality are (noting that $A^T=A$ and recalling that $V = V(0)$) \cite{Isaacs1965DG}:
\begin{eqnarray}
\dot{p} & = & -\frac{\partial }{\partial x}H \nonumber \\
& = & -k(x-\bar{x}) - Ap, \quad p(T) = 0 \label{eqn::ODEp1} \\
\dot{x} & = & Ax, \quad x(0) = x_0 \label{eqn::ODEx1} \\
u^\star(v) & = & \argmax_{U} H(x,p,u,v), \quad v^\star = \argmax_{V} H(x,p,u^\star(v),v).\nonumber
\end{eqnarray}
To prove the theorem, we will rely on approximating the state and costate up to first-order using Taylor expansion. To this end, we partition the problem's horizon into $L > K$ small sub-intervals of length $0<\delta\leq\tau$, where $\tau$ was defined in (\ref{eqn::dwellTime}), over which the system is time-invariant. More formally, define the times $0=t_1<t_2<\hdots<t_{L} < t_{L+1}=T$. Let $A_i$ be the system matrix corresponding to the interval $[t_i,t_{i+1}]$, $i=1,\hdots,L$. We will denote the $i$-th row of matrix $A_k$ by $A_{k,i}$ and its $(i,j)$-th element by $a^k_{ij}$. The proof comprises two steps: (i) we establish the claim of the theorem over $[t_{L},t_{L+1}]$; and (ii) we generalize the argument to hold over $[0,T]$. We start by considering the interval $[t_{L},t_{L+1}]$. The solutions to ODEs (\ref{eqn::ODEp1}) and (\ref{eqn::ODEx1}) over this interval are:
\begin{eqnarray}
x_{A_{L}}(t) & = & e^{A_L(t-t_L)}x_{A_L}(t_L) \\
p_{A_L}(t) & = & \int_t^T e^{-A_L(t-\tau)}(x_{A_L}(\tau)-\bar{x}) d\tau. \label{eqn::pDynNew}
\end{eqnarray}
Let $P_i(t) := e^{A_it} = I + tA_i + \bigO{\delta^2}$. We can then re-write the above expressions as
\begin{eqnarray*}
x_{A_L}(t) & = & P_L(t-t_L)x_{A_L}(t_L)\\
& = & (I + (t-t_L)A)x_{A_L}(t_L) + \bigO{\delta^2}\\
p_{A_L}(t) & = & \int_t^T P_L(\tau -t)[P_L(\tau-t_L)-M]x_{A_L}(t_L) d\tau\\
& = & (T-t)(I-M)x_{A_L}(t_L) + \bigO{\delta^2},
\end{eqnarray*}
where the last equality follows because $(T-t)(T-t_L)A_L = \bigO{\delta^2}$. Define $\xi(\alpha,\beta) := \alpha-\beta$, $\alpha,\beta \in \real$, and write
\begin{eqnarray}
x_{A_L}(t) & = & \left(I + \xi(t,t_L)A_L\right)x_{A_L}(t_L) + \bigO{\delta^2}\\
p_{A_L}(t) & = & \xi(T,t)(I-M)x_{A_L}(t_L) + \bigO{\delta^2}. \label{pExand1Int}
\end{eqnarray}
Further, define the matrices $G := I + \xi(t,t_L)A_L, R:=\xi(T,t)(I-M)$, and write
\begin{eqnarray*}
w_{ij} & = & a^L_{ij}(x_{A_L,i} -x_{A_L,j})(x_{A_L,j}-x_{A_L,i}) \\
& = & a^L_{ij}x_{A_L}(t_L)^T(G_i-G_j)(G_j-G_i)^Tx_{A_L}(t_L) + \bigO{\delta^2}\\
f_{ij} & = & a^L_{ij}x_{A_L}(t_L)^T(R_i-R_j)(G_j-G_i)^Tx_{A_L}(t_L)+ \bigO{\delta^2},
\end{eqnarray*}
where $R_i^T$, $R_i^T$ are the $i$-th row of $G$ and $R$, respectively. Using the above definitions, we obtain
\begin{eqnarray*}
(G_i-G_j)(G_j-G_i)^T & = & -(I_i-I_j)(I_i-I_j)^T - \xi(t,t_L)((I_i-I_j)(A_{L,i}-A_{L,j})^T\\
&&+(A_{L,i}-A_{L,j})(I_i-I_j)^T) - \xi(t,t_L)^2(A_{L,i}-A_{L,j})(A_{L,i}-A_{L,j})^T.
\end{eqnarray*}
The last term is quadratic, and thus we can absorb it in $\bigO{\delta^2}$. We then have
\begin{eqnarray*}
&& a^L_{ij}(G_i-G_j)(G_j-G_i)^T - a^L_{kl}(G_k-G_l)(G_l-G_k)^T  = a^L_{kl}(I_k-I_l)(I_k-I_l)^T\\
&&-a^L_{ij}(I_i-I_j)(I_i-I_j)^T + (a^L_{kl}(I_k-I_l)(A_{L,k}-A_{L,l})^T-a^L_{ij}(I_i-I_j)(A_{L,i}-A_{L,j})^T)\xi(t,t_L)\\
&& + (a^L_{kl}(A_{L,k}-A_{L,l})(I_k-I_l)^T-a^L_{ij}(A_{L,i}-A_{L,j})(I_i-I_j)^T)\xi(t,t_L) + \bigO{\delta^2}.
\end{eqnarray*}
Similarly, we have
\begin{eqnarray*}
&&a^L_{ij}(R_i-R_j)(G_j-G_i)^T - a^L_{kl}(R_k-R_l)(G_l-G_k)^T  = (a^L_{kl}(I_k-I_l)(I_k-I_l)^T\\
&& -a^L_{ij}(I_i-I_j)(I_i-I_j)^T)\xi(T,t)  + \bigO{\delta^2}.
\end{eqnarray*}
Let $\Gamma_1 = a^L_{kl}(I_k-I_l)(I_k-I_l)^T-a^L_{ij}(I_i-I_j)(I_i-I_j)^T$ and $\Gamma_2 = a^L_{kl}(I_k-I_l)(A_{L,k}-A_{L,l})^T-a^L_{ij}(I_i-I_j)(A_{L,i}-A_{L,j})^T$. We now have
\begin{eqnarray*}
w_{ij}-w_{kl} & = & x_{A_L}(t_L)^T(\Gamma_1 + \xi(t,t_L) \Gamma_2 + \xi(t,t_L) \Gamma_2^T)x_{A_L}(t_L) + \bigO{\delta^2}, \\
f_{ij}-f_{kl} & = & \xi(T,t)x_{A_L}(t_L)^T\Gamma_1x_{A_L}(t_L) + \bigO{\delta^2}.
\end{eqnarray*}
If $w_{ij}-w_{kl}\leq 0$, since $\xi(T,t) \geq 0$, we can write
\begin{equation*}
\xi(T,t)(w_{ij}-w_{kl}) = x_{A_L}(t_L)^T(\xi(T,t)\Gamma_1 + \xi(T,t)\xi(t,t_L) \Gamma_2 + \xi(T,t)\xi(t,t_L) \Gamma_2^T)x_{A_L}(t_L) + \bigO{\delta^2} \leq 0,
\end{equation*}
or $\xi(T,t)x_{A_L}(t_L)^T\Gamma_1x_{A_L}(t_L) + \bigO{\delta^2} \leq 0$, but the left hand side is $f_{ij}-f_{kl}$; hence, $w_{ij} \leq w_{kl} \implies f_{ij} \leq f_{kl}$ as required. So far, we have verified the claim of the theorem over the interval $[t_L,T]$ only. We are now in a position to generalize the statement of the theorem to the interval $[0,T]$. The only complication that arises when studying this interval is that the terminal condition, i.e. $p_{L-1}(t_L)$, is not forced to be zero as in $[t_L,T]$. Over the interval $[t_{L-1},t_{L}]$, the state and costate are
\begin{eqnarray*}
x_{L-1}(t) & = &e^{A_{L-1}(t-t_{L-1})}x_{A_{L-1}}(t_{L-1}) \\
p_{A_{L-1}}(t) & = & e^{-A_{L-1}(t-t_{L-1})}p_{A_{L-1}}(t_{L-1}) -\int_{t_{L-1}}^t e^{-A_{L-1}(t-\tau)}(x_{A_{L-1}}(\tau) -\bar{x}) d\tau.
\end{eqnarray*}
Solving for $p_{A_{L-1}}(t_{L-1})$ in terms of $p_{A_{L-1}}(t_L)$ and substituting back, we can write $p_{A_{L-1}}(t)$ in terms of $p_{A_{L-1}}(t_L)$ as follows:
\begin{equation} \label{eqn::pDynApprox}
p_{A_{L-1}}(t) = e^{-A_{L-1}(t-t_L)}p_{A_{L-1}}(t_L) +\int_t^{t_L} e^{-A_{L-1}(t-\tau)}(x_{A_{L-1}}(\tau) -\bar{x}) d\tau.
\end{equation}
By continuity of the state and costate functions, it follows that $x_{A_{L-1}}(t_L) = x_{A_L}(t_L)$, $p_{A_{L-1}}(t_L) = p_{A_L}(t_L)$. Using a first-order Taylor expansion and (\ref{pExand1Int}), we can write
\begin{eqnarray*}
p_{A_{L-1}}(t) & = & (I+\xi(t_L,t)A_{L-1})p_{A_L}(t_L) + \xi(t_L,t)(I-M)x_{A_{L-1}}(t_{L-1})+\bigO{\delta^2}\\
& = & \xi(t_L,t)(I-M)x_{A_{L-1}}(t_L) +  \xi(t_L,t)(I-M)x_{A_{L-1}}(t_{L-1})+\bigO{\delta^2}.
\end{eqnarray*}
We can further simplify this expression using $x_{A_{L-1}}(t)$ as follows:
\begin{eqnarray*}
\xi(t_L,t)(I-M)x_{A_{L-1}}(t_L) & = & \xi(t_L,t)(I-M)e^{A_{L-1}(t_L-t_{L-1})}x_{A_{L-1}}(t_{L-1}) \\
& = &\xi(t_L,t)(I-M)x_{A_{L-1}}(t_{L-1}) + \bigO{\delta^2},
\end{eqnarray*}
and therefore we have
\begin{equation} \label{pExand2Int}
p_{A_{L-1}}(t) = 2\xi(t_L,t)(I-M)x_{A_{L-1}}(t_{L-1}) + \bigO{\delta^2}.
\end{equation}
Comparing (\ref{pExand1Int}) and (\ref{pExand2Int}), we conclude that the argument used to prove the claim over the interval $[t_L,T]$ applies over $[t_{L-1},t_L]$. Hence, $w_{ij}-w_{kl}  \leq 0$ implies that $f_{ij}-f_{kl} \leq 0$ over $[t_{L-1},t_L]$.

Note that we can generalize (\ref{eqn::pDynApprox}) to any interval $[t_i,t_{i+1}]$, $i=1,\hdots,L$, as follows:
\[
p_{A_i}(t) = e^{-A_i(t-t_{i+1})}p_{A_i}(t_{i+1}) +\int_t^{t_{i+1}} e^{-A_i(t-\tau)}(x_{A_i}(\tau) -\bar{x}) d\tau.
\]
Following similar steps to the above, we can arrive at
\[
p_{A_i}(t) = \frac{T-t_i}{\delta}\xi(t_{i+1},t)(I-M)x_{A_i}(t_i)+\bigO{\delta^2}, \quad t\in[t_i,t_{i+1}],
\]
which maintains the same structure as in (\ref{pExand2Int}), and the claim therefore holds for the interval $[t_i,t_{i+1}]$, $i=1,\hdots,L$, and the theorem is proved.

\subsection{Technical Results}
{
\begin{proposition} \label{prop::bigO}
Given $\tau_1,\tau_2,\tau_3$, which were defined in terms of $\delta>0$ in Theorem \ref{thm::implement}, let $f$ be a real-valued function. Then, if $f(\delta) = \bigO{\tau_i^2}$ as $\delta \to 0$, we have $f(\delta) = \bigO{\delta^2},$ $i \in\{1,2,3\}$. Also, if $f(\delta) = \tau_i\bigO{\tau_j^2}$ as $\delta \to 0$, then $f(\delta) = \bigO{\delta^3},$ $i,j \in\{1,2,3\}$.
\end{proposition}
\begin{proof}
Recall that we write $f(x) = \bigO{g(x)}$, for some real-valued function $g$, as $x\to a$ if there exist constants $M,\gamma$ such that $|f(x)| \leq M|g(x)|$, for all $x$ satisfying $|x-a| < \gamma$. Since $f(\delta) = \bigO{\tau_i^2}$ as $\delta \to 0$, and recalling that by definition we have $\tau_i \leq \delta$ for $i\in \{1,2,3\}$, we can write $f(\delta) \leq M\tau_i^2 \leq M \delta^2$. Hence, $f(\delta) = \bigO{\delta^2}$. To prove the second statement, recall that $h(x)\bigO{g(x)} = \bigO{h(x)g(x)}$, for any two real-valued functions $h,g$. Hence, as $\delta \to 0$, we have $f(\delta) = \tau_i\bigO{\tau_j^2} = \bigO{\tau_i \tau_j^2}$. Therefore, $f(\delta) \leq M\tau_i \tau_j^2\leq M\delta^3$ and $f(\delta) = \bigO{\delta^3}$.
\end{proof}
}
\begin{lemma} \label{lemma::upperLowerBounds}
Given $\epsilon > 0$ and $\delta \leq \tau$, $\tau$ defined in (\ref{eqn::dwellTime}), one can select the problem horizon $T$ small enough such that
\begin{equation} \label{eqn::upperLowerBoudns}
\epsilon \leq |x_i(t)-x_j(t)| \leq 2\vnorm{x_0}_{\infty}, \quad \forall e_{ij} \in \mE,
\end{equation}
for all $t\in [0,T]$.
\end{lemma}
\begin{proof}
By the structure of the system matrix in (\ref{systemEqn}), we can deduce that $|x_i-x_j|$ cannot increase as $t\to T$. Thus
\begin{eqnarray*}
|x_i(t)-x_j(t)| & \leq & \max_{1\leq i,j \leq n} |x_i(0) -x_j(0)| \\
& \leq & 2\max_{1\leq i \leq n}|x_i(0)| = 2 \vnorm{x_0}_{\infty}.
\end{eqnarray*}
This provides the uniform upper bound. In order to obtain a uniform lower bound, we need to ensure that $|x_i(t)-x_j(t)|$ does not approach zero as $t\to T$. We are seeking a time $t^\star$ such that for a given $\epsilon>0$, we have $|x_i(t)-x_j(t)| \geq \epsilon$ for all $t < t^\star$ and all $e_{ij} \in \mE$. We can then fix $T<t^\star$ to ensure the existence of a uniform lower bound on $|x_i(t)-x_j(t)|$. Let us again restrict our attention to a small interval $[t_0,t_0+\delta]$ where the system is time-invariant, and let the system matrix over this interval be $A$. We require that the system did not reach equilibrium over this interval, i.e., $x(t_0+\delta) \neq \bar{x}$. Without loss of generality, we assume that $x_1(t_0) > \hdots > x_n(t_0)$\footnote{We are making the implicit assumption that $x_1(0) > \hdots > x_n(0)$.}. Define the following dynamics
\begin{eqnarray*}
\frac{d}{dt}(\overline{y}_i - x_1(t_0))  =  \sum_{j\neq i}A_{ij}(x_1(t_0)-\overline{y}_i), \quad
\frac{d}{dt}(\underline{y}_i-x_n(t_0)) = \sum_{j\neq i}A_{ij}(x_n(t_0)-\underline{y}_i),
\end{eqnarray*}
with initial conditions $\overline{y}_i(t_0) = 2x_1(t_0)$, $\underline{y}_i(t_0) = 2x_n(t_0)$. Note that $\dot{x}_i = \sum_{j\neq i}A_{ij}(x_j-x_i)$. It follows that $\dot{\underline{y}}_i \leq \dot{x}_i \leq \dot{\overline{y}}_i$. By the comparison principle, we conclude that $\underline{y}_i - x_n(t_0) \leq x_i \leq \overline{y}_i-x_1(t_0)$, for $i \in \mN$. Note that we can readily find the solution trajectories for $\overline{y}$ and $\underline{y}$. By defining $a_i = \sum_{j\neq i} A_{ij}$, we can then write
\begin{eqnarray*}
\overline{y}_i - x_1(t_0) = e^{-a_i(t-t_0)}x_1(t_0) , \quad \underline{y}_i - x_n(t_0) = e^{-a_i(t-t_0)}x_n(t_0).
\end{eqnarray*}
By solving the equation $\overline{y}_{i-i}-x_1(t_0)=\underline{y}_i-x_n(t_0)$, we can find a time $t^\star_i$ when $x_{i-1}$ can potentially meet $x_i$:
\[
t_i^\star = \frac{1}{a_{i-1}-a_i}\ln\left(\frac{x_1(t_0)}{x_n(t_0)}\right) + t_0.
\]
If $t^\star_i > t_0 + \delta$, for all $i\in \mN$, then we need to propagate the solution forward, and keeping in mind that the system matrix could change, until we find a time $t^\star_i$ in some interval $[\tilde{t},\tilde{t}+\delta]$ where $\overline{y}_{i-i}=\underline{y}_i$ for some $i\in \mN$. Then, for a given $\epsilon>0$, we can select $T < t_i^\star$ such that $|x_i-x_{i-1}|\geq |\underline{y}_i-\overline{y}_{i-1}| \geq \epsilon$; hence, we conclude that for this choice of $T$ we can guarantee that $|x_i-x_j|\geq \epsilon>0$ for all $e_{ij} \in \mE$.
\end{proof}

\end{document}